\documentclass[fleqn,letterpaper,UKenglish,11pt]{article}
\usepackage[margin=1.6in]{geometry}

\usepackage[l2tabu, orthodox]{nag}
\usepackage[utf8]{inputenc}
\usepackage[pdfusetitle]{hyperref}
\usepackage{microtype}
\usepackage{tcolorbox}
\usepackage{xcolor}

\usepackage{mathtools}

\usepackage{amsmath,amssymb,amsfonts,amsthm}
\theoremstyle{plain}

\newtheorem{theorem}{Theorem}[section]

\newtheorem*{conjecture*}{Conjecture}
\newtheorem{proposition}[theorem]{Proposition}
\newtheorem{lemma}[theorem]{Lemma}

\newtheorem{corollary}[theorem]{Corollary}

\theoremstyle{definition}
\newtheorem{definition}[theorem]{Definition}
\newtheorem{example}[theorem]{Example}

\usepackage{enumerate} 
\usepackage{authblk}
\usepackage{xspace}
\usepackage{todonotes}
\usepackage{booktabs}
\usepackage{graphicx}
\usepackage{multicol}
\usepackage{algorithm}
\usepackage{algorithmic}
\usepackage{cleveref}

\usepackage{tikz}
\usetikzlibrary{matrix,calc,arrows,positioning,graphs}

\newcommand{\cC}{\mathcal{C}}
\newcommand{\cQ}{\mathcal{Q}}

\newcommand{\NN}{\mathbb{N}}

\newcommand{\overbar}[1]{\mkern 0.8mu\overline{\mkern-0.8mu#1\mkern-0.8mu}\mkern 0.8mu}

\newcommand{\XP}{XP}
\newcommand{\PSPACE}{PSPACE\xspace}
\newcommand{\bigoh}{\mathcal{O}}

\newcommand{\vars}{\ensuremath{V}}

\newcommand{\pref}{\ensuremath{\mathcal{Q}}}

\newcommand{\xor}{\oplus}

\newcommand{\repl}[1]{[#1]}
\newcommand{\CCC}{\mathcal{C}}

\newcommand{\TCNF}{\textup{2-CNF}\xspace}
\newcommand{\DCNF}{\textup{$d$-CNF}\xspace}
\newcommand{\CNF}{\textup{CNF}\xspace}

\newcommand{\AFF}{\textup{AFF}\xspace}
\newcommand{\HORN}{\textup{HORN}\xspace}
\newcommand{\CNFAFF}{\textup{CNF+AFF}\xspace}
\newcommand{\DUALHORN}{\textup{DualHORN}\xspace}

\newcommand{\EN}{\ensuremath{\textup{IHSB}_{-}}\xspace}
\newcommand{\EP}{\ensuremath{\textup{IHSB}_{+}}\xspace}

\newcommand{\QBF}{\textsc{QBF}\xspace}
\newcommand{\QCSP}{\textsc{QBF}}

\newcommand{\pol}{\mathrm{Pol}}

\DeclareMathOperator{\maj}{majority}
\DeclareMathOperator{\mnrty}{minority}

\newcommand{\kernel}{\mathrm{Kernel}}
\newcommand{\pivot}{\mathrm{pivot}}
\newcommand{\elim}{\mathrm{elim}}

\newcommand{\cnfprop}[1]{{\texttt{prop}}({#1})}

\newcommand{\dom}{\mathrm{D}}

\renewcommand{\implies}{\Rightarrow}

\crefname{theorem}{theorem}{theorems}
\Crefname{theorem}{Theorem}{Theorems}
\crefname{lemma}{lemma}{lemmas}
\Crefname{lemma}{Lemma}{Lemmas}
\crefname{definition}{definition}{definitions}
\Crefname{definition}{Definition}{Definitions}

\newif\iflong
\newif\ifshort

\longtrue
    
\iflong
\else
\shorttrue
\fi

\begin{document}

\title{Clausal Deletion Backdoors for QBF: a Parameterized Complexity Approach}

\author[1]{Leif Eriksson}
\author[2]{Victor Lagerkvist}
\author[3]{Sebastian Ordyniak}
\author[2,4]{George Osipov}
\author[3]{Fahad Panolan}
\author[1]{Mateusz Rychlicki}

\affil[1]{Independent researcher}
\affil[2]{Link{\"o}ping University, Sweden}
\affil[3]{University of Leeds, UK}
\affil[4]{Royal Holloway, University of London, UK}

\date{}
\maketitle

\begin{abstract}
Determining the validity of a quantified Boolean formula (QBF) is a \PSPACE-complete problem with rich expressive power. Despite interest in efficient solvers, there is, compared to problems in NP, a lack of positive theoretical results, and in the parameterized complexity setting one often has to restrict the quantifier prefix (e.g., bounding alternations) to obtain fixed parameter tractability (FPT). We propose a new parameter: the number of variables in clauses that has to be removed before reaching a tractable class (a {\em clause covering (CC) backdoor}). We are then interested in solving QBF in FPT time given a CC-backdoor of size $k$. We consider the three classical, tractable cases of QBF as base classes: \emph{Horn}, \emph{2-CNF}, and \emph{linear equations}. We establish $\mathrm{W}[1]$-hardness for Horn but prove FPT for the others, and prove that in a precise, algebraic sense, we are only missing one important case for a full dichotomy. Our algorithms are non-trivial and depend on propagation, and Gaussian elimination, respectively, and are comparably unexplored for QBF.
\end{abstract}

\section{Introduction}

The {\em Boolean satisfiability problem} (SAT) is the canonical NP-complete problem. Many hard problems can today be solved efficiently by a suitable SAT encoding together with off-the-shelf solvers~\cite{biere2021handbook}. While this often works well for problems in NP, many fundamental problems in AI and machine learning are {\em harder} than NP, with \PSPACE crystallizing as the most fundamental class. Examples include neural network training \cite{ml2}, graph neural network satisfiability \cite{ml4}, and quantized graph neural network verification \cite{ml5}, all of which are \PSPACE-complete or contained in \PSPACE. In symbolic AI, planning problems like STRIPS and SAS$^+$ also fall inside \PSPACE. 

Solving \PSPACE problems is thus a pressing issue. The canonical hard problem is the {\em quantified Boolean formula} problem (QBF), deciding the validity of $\Psi = Q_1 x_1 \ldots Q_n x_n . \phi(x_1,\ldots,x_n)$ where $Q_i \in \{\forall, \exists\}$ and $\phi$ is in \CNF. Unfortunately, we lack fast techniques for solving QBF instances. Theoretically, $\Psi$ can be solved in $2^n$ time by branching: $\exists$ requires one successful branch while $\forall$ requires both. While SAT's $2^n$ complexity can be improved to $c_k^n$ for $c_k < 2$ for any $k$-\CNF, no such improvement is known even for $k = 3$ for QBF, and it is conjectured impossible unless the {\em strong exponential-time hypothesis} (SETH) \cite{impagliazzo2009} fails \cite{CalabroIP13}.

Hence, we move to {\em parameterized complexity} to identify
parameters $k$ (corresponding to a ``hidden structure'') leading to
tractability. We seek {\em fixed-parameter tractability}
($\mathrm{FPT}$), i.e., $f(k) \cdot ||\Psi||^{\bigoh(1)}$ time for a
computable $f$. $\mathrm{FPT}$ is also characterized by a
\emph{kernel}: a polynomial-time procedure reducing an instance
$(\Psi, k)$ to $(\Psi', k')$ where $||\Psi'|| + k' \leq f(k)$. A
well-known parameter is {\em treewidth}, measuring how tree-like a
formula is. 
  For instance, while $\QBF$ remains \PSPACE-complete even
  for constant treewidth~\cite{ATSERIAS20141415}
  it is $\mathrm{FPT}$
  parameterized by treewidth
  plus quantifier alternations. However, even in this case, $f$ is a
  {\em power tower} which, even for small values, requires more
  operations than there are atoms in the universe.

The most successful parameter besides treewidth is arguably {\em backdoor size}, which measures how close an instance is to a tractable class. For SAT, a (strong) backdoor for $\phi$ is a set $B \subseteq V(\phi)$ such that $\phi[\tau]$ belongs to a fixed, tractable class (e.g., $2$-\CNF) for each partial assignment $\tau \colon B \to \{0,1\}$. Here, $\phi[\tau]$ is the formula obtained by replacing variables in $B$ with their values under $\tau$ and simplifying, and $V(\phi)$ is the set of variables in $\phi$. However, backdoors cannot readily be applied to QBF since branching on $B$ may violate the quantifier prefix. This can be circumvented with restrictions on the quantifier prefix~\cite{SamerSzeider07a} (e.g., that backdoor variables are to the left), but this only provides a meaningful speed-up for a very narrow class of QBF formulas.

In this paper we investigate an alternative QBF backdoor notion which avoids any such restrictions. We primarily consider 
$\TCNF$, $\HORN/\DUALHORN$ (at most one positive/negative literal in each clause), and $\AFF$ (systems of linear equations over GF(2)).
We assume that the formulas in each class are conjunctions of clauses/atoms, e.g., each formula in $\TCNF$ is a conjunction of clauses of arity 2. For a class of formulas $\cC$ we write $\QBF(\cC)$ for the $\QBF$ problem restricted to instances in $\cC$.
A first idea could be to measure the distance to a tractable class $\CCC$ with respect to the number of \emph{clauses}  that are outside $\CCC$. I.e., if we view $\phi$ as a set of clauses, we let $\phi_{\CCC} \subseteq \phi$ be the maximal subset of clauses in $\phi$ such that $\phi_{\CCC}$ belongs to $\CCC$, and parameterize by $|\phi \setminus \phi_{\CCC}|$. %
Unfortunately, allowing even one out-of-class clause to be added to a 2-SAT or a Horn formula makes the problem \PSPACE-complete~\cite{DBLP:conf/lics/FichteGHSO23}.
Thus, we propose parameterization by $|V(\phi \setminus \phi_{\CCC})|$, i.e., the number of variables that cover all out-of-class clauses, leading to the following definition.

\begin{definition}
  Let $\mathcal{C}$ be a class of formulas. Let $\Phi=\cQ.\phi$ be a
   $\QBF$ formula. A set $B$ of variables
  is a \emph{clause covering backdoor} (CC-backdoor) for $\Phi$ if
  $\phi\setminus C \in \cC$ for a subset $C \subseteq \phi$
  of clauses such that $V(C) = B$.
\end{definition}  

Note that the {\em CC-backdoor detection problem} is trivial for all
classical base classes: if the tractable class (e.g., \TCNF) is fixed
then we simply collect all out-of-class clauses, compute the set of
variables covering them, and see whether it has the required
size. Hence, the computationally interesting problem is {\em
  CC-backdoor evaluation}: given a QBF formula $\Psi$ and CC-backdoor
of size $k$, can $\Psi$ be evaluated in $\mathrm{FPT}$ time with
respect to $k$?

Given the lack of
$\mathrm{FPT}$ results for QBF with unbounded quantifier alternations, there is
perhaps not much reason to be optimistic.
Furthermore, the standard quantifier elimination does not seem directly compatible with
CC-backdoor evaluation problem: by eliminating backdoor variables we get a complicated disjunctive formula which does not appear to be easier to solve than the original formula.
The main complication, then, is that the backdoor variables may be involved in the in-the-class clauses, i.e., some of them may be in $V(\phi_{\CCC})$, so it is not possible to solve the backdoor part independently.
Moreover, they can occur anywhere in the quantifier prefix.

Despite this grim outlook we manage to prove $\mathrm{FPT}$ for \TCNF
and \AFF. For \TCNF we use an elaborate branching strategy with
look-ahead. Essentially, by going from left to right in the quantifier
prefix we prove that one either (1) can branch and always ``hits'' a
variable in the backdoor and therefore decreases the parameter $k$, or (2) if branching would not affect the backdoor variables, then one can disregard one of the branches. Then, depending on the current quantifier ($\exists$ or $\forall$) we either continue the process or correctly report unsatisfiability. Concurrently, we repeatedly propagate the input formula to ensure that we do not miss any immediate implications. Put together, this results in an $\mathrm{FPT}$ algorithm where $f(k) = 2^k$ and with a polynomial factor dominated by $n^3$ from propagation.
Under SETH, the dependence on $k$ cannot be improved: in the special case when all clauses are out-of-class, the CC-backdoor evaluation problem is as hard as solving an arbitrary QBF instance, which cannot be solved substantially faster than $2^{n}$ time under SETH.
We stress that, a priori, it is rather surprising that non-trivial branching is applicable: while the algorithm can be stated in a rather succinct form it is non-trivial to prove correctness.

The affine case, however, cannot be solved in $\mathrm{FPT}$-time by branching in the order of quantifier prefix, 
and requires a drastically different strategy.
For this case it is conceptually easier to construct a kernel instead of
deriving an $\mathrm{FPT}$ algorithm. Naturally, the main
complication is that our input instance $\Psi = Q_1 x_1 \ldots Q_n
x_n . \phi$ contains both affine equations (easily solvable by
Gaussian elimination) and arbitrary clauses (where Gaussian
elimination is not applicable). To handle this we isolate the
influence of backdoor variables by a carefully orchestrated sequence
of choosing which variables to pivot on, and 
finally reach the desired conclusion where at most $2k$ variables remain. 
By branching on these variables and observing that at most
$k$ of them are not determined by previous choices,
i.e., only $k$ need to be branched upon,
we can then solve the problem in $\bigoh^*(2^k)$ time, 
and the running time cannot be improved under the SETH.

Inspired by these two positive results it is natural to consider to which extent our classification is complete, i.e., are there more FPT cases among the natural tractable classes of QBF?
We begin studying this in Section~\ref{sec:cc_horn} where we as a first step prove a negative result for $\HORN$ (and $\DUALHORN$): CC-backdoor evaluation is {\em not} $\mathrm{FPT}$ unless a widely believed conjecture in parameterized complexity is false ($\mathrm{FPT} \neq \mathrm{W}[1]$). This does not cover everything, however, and we manage to extend W[1]-hardness to 
the class of \emph{implicative hitting set-bounded} formulas $\EN \subseteq \HORN$ (negative clauses of any length, and binary implication) and its dual $\EP \subseteq \DUALHORN$ (positive clauses of any length, and binary implication). It is easy to show that the problem becomes FPT if we remove implication. However, in contrast, keeping implication and all positive (or negative) clauses of arity at most $d$ ($d$-$\EP$, $d$-$\EN$) results in a problem where it seems much more challenging to obtain FPT.
With this troublesome case in mind we investigate (in Section~\ref{sec:qcsp}) exactly how far away we are from a complete CC-evaluation dichotomy (FPT versus W[1]-hard). We do this not only for all possible sets of backdoor clauses, but for the extension of QBF to arbitrary constraints, the \emph{quantified constraint satisfaction problem} (QCSP). Here, we show that the so-called \emph{algebraic approach} is compatible with CC-backdoors, which greatly reduces the number of cases that we need to consider. With this machinery we manage to show that the CC-evaluation problem is always either FPT or W[1]-hard \emph{except} for $d$-$\EP$ (and $d$-$\EN$) for $d \geq 3$ where the complexity status is unknown.

Put together, we thus manage to solve arbitrary QBF instances in optimal time with respect to the number of variables in the CC-backdoor, {\em without} any assumptions on the quantifier prefix. We accomplished this with two different algorithmic schemes, branching and Gaussian elimination, and the latter is overlooked in the QBF literature. CC-backdoors may thus open up the usage of algorithmic schemes that do not work for the full QBF problem, while still being reasonably expressive since we do not have to make any assumptions on the quantifier prefix. Hence, we have obtained a novel paradigm for solving QBF, and in the long term, possibly more \PSPACE-complete problems. We discuss this and other open questions in Section~\ref{sec:discussion}.

\ifshort
\smallskip
\noindent {\emph{Statements whose full proofs are omitted due to
space constraints are marked with $\star$ and the full proofs can be
found in the supplementary material.}}
\fi

\section{Preliminaries} \label{sec:preliminaries}
We begin by formally introducing the quantified Boolean formula problem and some additional notation.
\subsection{Clauses and Relations}

Boolean expressions, formulas, variables, 
literals, conjunctive normal form (\CNF), clauses and atoms are
defined in the standard way~(cf.~\cite{biere2021handbook}).
We treat $1$/$\top$ and $0/\bot$ as the truth values
``true'' and ``false'', respectively, 
and formulas in \CNF as sets of clauses where each clause is a set of literals. 
More generally, it is convenient to use a relational perspective where atoms are expressions of the form $R(x_1,  \ldots, x_r)$ where $x_1, \ldots, x_r$ are variables and $R \subseteq \{0,1\}^r$ is an $r$-ary relation. 
A function $f \colon \{x_1, \ldots, x_r\} \to \{0,1\}$ satisfies an atom $R(x_1, \ldots, x_r)$ if $(f(x_1), \ldots, f(x_r)) \in R$. Note that clauses can be defined in this way, e.g., the clause $(x \lor y \lor \lnot z)$ is equivalent to an atom $R(x,y,z)$ where $R = \{0,1\}^3 \setminus \{(0,0,1)\}$.

Given a Boolean formula $\phi$, we write $\vars(\phi)$ for the set of variables in $\phi$. 
For a (partial) assignment $\tau : V' \rightarrow \{0,1\}$, where $V'\subseteq V(\phi)$, 
and an atom $R(x_1, \ldots, x_r)$, 
we let $R(x_1, \ldots, x_r)[\tau]$ be the atom obtained by (1) removing any tuple $(b_1, \ldots, b_r) \in R$ if there is an $1 \leq i \leq r$ where $x_i \in V'$ and $\tau(x_i) \neq b_i$, and then (2) projecting away any $x_i$ such that $x_i \in V'$. 
This easily extends to conjunctions of atoms, and for a formula $\phi$ we thus let 
$\phi[\tau]$ be the formula obtained by simplifying each atom according to $\tau$.
If $V' = \{x\}$ and $\tau(x) = b$, then we simply write $\phi\repl{x = b}$.
In the special case when $\phi$ is in \CNF, $\phi[\tau]$ is the formula obtained from $\phi$ by removing all clauses 
satisfied by $\tau$ and removing all literals from clauses that are set to $0$ by $\tau$. 

Let us now define the following classes of conjunctive Boolean formulas.

\begin{enumerate}
\item
 \DCNF for the class of formulas in \CNF where each clause has at most $d$ literals.
\item
 \HORN for the class of formulas in \CNF where each clause has at most one positive literal.
 \item $\EN \subseteq \HORN$ restricted to purely negative clauses, unit clauses, and implications $(x \implies y)$, and $d$-$\EN \subseteq \EN$ for the restriction to at most $d$-ary clauses.
\item
\AFF for the class of conjunctive formulas with atoms of the form
  $(\ell_1 \oplus \ell_2 \oplus \dots \oplus \ell_k) = b$ where $b
  \in \{0,1\}$ and each $\ell_i = b_i x_i$ for a variable $x_i$ and
  coefficient $b_i \in \{0,1\}$. We 
  represent the atoms of affine formulas by a pair $(A,b)$, where $A$ is
  a set of variables and $b \in \{0,1\}$. 
  We refer to such atoms as \emph{equations}.
  \item 
  If $\CCC$ and $\CCC'$ are both classes of formulas then we write $\CCC + \CCC'$ for the class of formulas where each atom is in $\CCC \cup \CCC'$.
  \CNFAFF for the class of conjunctive formulas where each atom is either in $\CNF$ or in $\AFF$.
\end{enumerate}
We sometimes also consider dual cases and write $\DUALHORN$,  $\EP$  and $d$-$\EP$ for the classes obtained by flipping the literals of $\HORN$, $\EN$, and $d$-$\EN$. We say that a conjunctive formula is in \emph{clausal form} if all
its atoms are clauses.

\begin{example}
Consider the propositional formula 
\[ \phi = (x_1 \lor x_2 \lor x_3) \land (x_1 \oplus x_2 = 1)\] 
in $\CNFAFF$. 
Here $(x_1 \lor x_2 \lor x_3) \in \CNF$ and $(x_1 \oplus x_2 = 1) \in \AFF$. In the relational perspective we view $\phi$ as $R(x_1, x_2, x_3) \land R'(x_1, x_2)$ where $R = \{0,1\}^{3} \setminus \{(0,0,0)\}$ and $R' = \{(0,1), (1,0)\}$.
\end{example}

The {\em satisfiability} problem (SAT) for $\phi$ is then simply to determine whether there exists a satisfying assignment $f \colon \vars(\phi) \to \{0,1\}$, and by restricting $\phi$ to formulas in $d$-\CNF we obtain the well-known {\em $d$-SAT} problem.

When working with \TCNF it will be convenient to assume that the set of clauses is maximally propagated in the sense that no further clauses can be inferred via resolution.
For example, the formula $(\lnot x \lor y) \land (\lnot y \lor z)$ is not propagated because
it also implies $(\lnot x \lor z)$. 
Adding the latter makes the formula propagated.
We formalize this as follows:
let $C_1 = (\ell \lor A)$ and $C_2 = (\lnot \ell \lor B)$ be two
clauses in a \CNF formula, where $\ell$ is a literal. The {\em
  resolution} of $C_1$ and $C_2$ derives the {\em resolvent} clause $C = (A \lor B)$. 
We say that a $\TCNF$ formula $\phi$ is {\em propagated} if 
(1) $\phi = \{\bot\}$ only if $\bot \in \phi$ or $\emptyset \in \phi$\footnote{Recall that $\bot$ is the constant 0, while $\emptyset$ is the empty clause, logically equivalent to $\bot$.},
(2) $\phi$ is closed under resolution, and 
(3) no clause is a super-clause of another clause. 
Let $\cnfprop{\phi}$ denote the \TCNF formula $\phi'$ such that 
$\phi'$ is propagated and $\phi \equiv \phi'$. 
It is well-known that given a \TCNF formula $\phi$,
the propagated formula $\phi'$ can be computed in polynomial time
(see e.g~\cite{biere2021handbook}).

\begin{proposition} \label[prop]{prop:prop_comp}
Let $\phi$ be a \TCNF formula. For any \TCNF formula $\phi$, 
$\cnfprop{\phi}$ can be computed in $\bigoh(|V(\phi)|^3)$ time, 
and $\cnfprop{\phi}$ has at most $|V(\phi)|^2$ clauses.
\end{proposition}

\subsection{The Quantified Boolean Formula Problem}

A \emph{quantified Boolean formula (QBF)} 
is of the form $\cQ. \phi$, where 
$\cQ = Q_1 x_1 \ldots Q_n x_n$ with
$Q_i \in \{\forall, \exists\}$ for all $1 \leq i \leq n$
is the \emph{(quantifier) prefix},
$x_1, \dots, x_n$ are variables, and
$\phi$ is a Boolean formula 
on these variables
called the
\emph{matrix}. For any set $S=\{x_i,\ldots,x_j\}\subseteq V(\phi)$ the variable $x_i$ said to be the {\em outermost} variable if no other variable in $S$ occurs before $x_i$ in $\pref$ and similarly $x_n$ the {\em innermost} if no other variable in $S$ occurs after $x_i$ in $\pref$. When speaking of the entire quantifier prefix without any specific subset then $x_1$ is the outermost and $x_n$ the innermost.
If $Q_i = \forall$, we say that
$x_i$ is a \emph{universal variable},
and if $Q_i = \exists$, we say that
$x_i$ is an \emph{existential variable}.

The truth value of a QBF $\cQ. \phi$ is 
defined recursively.
Let $\cQ = Q_1 x_1 \dots Q_n x_n$ and
$\cQ' = Q_2 x_2 \dots Q_n x_n$.
Then $\cQ. \phi$ is \emph{true} if
\begin{itemize}
  \item $\phi$ is a true Boolean expression, or
  \item $Q_1 = \exists$ and 
  $\cQ'. \phi\repl{x_1 = 0}$ \emph{or} $\cQ'. \phi\repl{x_1 = 1}$ is true, or
  \item $Q_1 = \forall$ and 
  $\cQ'. \phi\repl{x_1 = 0}$ \emph{and} $\cQ'. \phi\repl{x_1 = 1}$ are true.
\end{itemize}
Otherwise, $\cQ. \phi$ is \emph{false}. 
For a QBF $\Phi=\pref.\phi$ and a partial assignment $\tau : V' \rightarrow \{0,1\}$, where $V'\subseteq V(\Phi)$, we write $\Phi[\tau]$ for the QBF formula $\pref'.\phi[\tau]$, where $\pref'$ is obtained from $\pref$ after removing the variables in $V'$.

\begin{definition}
For a class of formulas $\CCC$ we write $\QBF(\CCC)$ 
for the computational problem of deciding whether a QBF formula 
$\cQ . \phi$, where $\phi \in \CCC$, is true.
We may also refer to an instance of 
this problem as a $\QBF(\CCC)$ formula.
\end{definition}

The evaluation  of a QBF  $Q_1 x_1 \dots Q_n x_n . \phi$ can be seen as a two-player game between the {\em universal} and {\em existential} players. In the $i$th step the player $Q_i$ assigns value to the variable $x_i$. The existential player wins the game if $\phi$ evaluates
to true under the assignment constructed in the game and the universal player wins
if $\phi$ evaluates to false.

A \emph{strategy} for a player is a rooted binary tree \( T \) of
depth \( n+1 \) where nodes at level \( j \in \{0, \ldots, n-1\} \)
are labeled \( x_{j+1} \), leaves are labeled \( \top \) or \( \bot
\), and edges are labeled \( 0 \) or \( 1 \). If $T$ is an existential
strategy, nodes labeled by existential variables have one child, and
those by universal variables have two; the roles are reversed for
universal strategies. Each node \( t \) corresponds to the assignment
\( \tau_t^T \) given by the vertex and edge labels on the root-to-\( t
\) path, defining \( \phi[\tau_l^T] \) as the label of each leaf \( l
\). An existential strategy is \emph{winning} if all leaves are \( \top \),
and a universal strategy is \emph{winning} if all leaves are \( \bot \).

\begin{proposition}
A QBF formula $\Phi$ is true ($\Phi \Leftrightarrow \top$) if and only if the existential player has a winning strategy, and $\Phi$ is false ($\Phi \Leftrightarrow \bot$) if and only if 
the universal player has a winning strategy.     
\end{proposition}

\begin{example} \label{ex:tcnf_backdoor}
Consider the QBF formula
\[
  \Phi = \exists x_1\, \forall x_2\, \exists x_3\, \exists x_4\, \exists x_5\,.\,(\phi_1 \land \phi_2),
\]
where
\[
  \phi_1 = (x_1 \lor x_3) \land (\neg x_1 \lor x_4) \land (x_3 \lor x_4) \land (x_2 \lor x_5)
\]
is in $\TCNF$ (and is propagated) and
\[
  \phi_2 = (\neg x_3 \lor \neg x_4 \lor \neg x_5).
\]
Then $V(\phi_2) = \{x_3, x_4, x_5\}$ is a CC-backdoor of size $3$ into $\TCNF$. Note that the backdoor variables also appear in $\phi_1$ and lie in the inner part of the prefix, behind the universal~$x_2$, so they cannot be branched on directly without first handling earlier quantifiers.
\end{example}

\subsection{Parameterized Complexity}

We use standard notation and primarily follow~\cite{downey2013fundamentals}~and~\cite{fomin2019kernelization}.
For a finite alphabet $\Sigma$ a 
{\em parameterized problem} $L$ is a subset of $\Sigma^* \times \NN$.
The problem $L$ is \emph{fixed-parameter tractable} (or, in $\mathrm{FPT}$)
if there is an algorithm deciding 
whether an instance $(I, k) \in \Sigma^* \times \NN$ is in $L$
in time $f(k) \cdot |I|^c$, where
$f$ is some computable function and
$c$ is a constant independent of $(I, k)$. 
To state $\mathrm{FPT}$ results we sometimes use the $\bigoh^*$ notation which suppresses polynomial factors in the input size, i.e., we simply write $\bigoh^*(f(k))$ rather than $f(k) \cdot |I|^c$.
An equivalent definition of $\mathrm{FPT}$ can be given by a \emph{kernelization (algorithm)} for $L$, i.e.,
an algorithm that takes $(I, k) \in \Sigma^* \times \NN$ as input and
in time polynomial in $|I|+k$, 
outputs $(I',k') \in \Sigma^* \times \NN$ such that $(I,k)
  \in L$ if and only if $(I',k')\in L$, and $|I'|,k' \leq h(k)$ for
  some computable function $h$. 

Let $L, L' \subseteq \Sigma^* \times \NN$ be two parameterized problems.
A mapping $P : \Sigma^* \times \NN \to \Sigma^* \times \NN$ is a 
\emph{fixed-parameter ($\mathrm{FPT}$) reduction from $L$ to $L'$}
if there exist computable functions $f,p : \NN \to \NN$
and a constant $c$ such that (1) $(I,k) \in L$ if and only if $P(I,k) = (I',k') \in L'$,
(2) $k' \leq p(k)$, and
(3) $P(I,k)$ can be computed in $f(k) \cdot |I|^c$ time.
We write $L \leq_{\mathrm{FPT}} L'$ if this holds, and if, additionally, $L' \leq_{\mathrm{FPT}} L$ then $L$ and $L'$ are said to be \emph{FPT-equivalent}.
Parameterized complexity also contains a complementary theory of hardness. Here, the hard classes $\mathrm{W}[1]\subseteq \mathrm{W}[2] \subseteq \ldots $, form a hierarchy of classes. We say that a problem is {\em $\mathrm{W}[1]$-hard} if it admits an $\mathrm{FPT}$-reduction from \textsc{Independent Set}
(parameterized by the number of vertices in the independent set).
For sharper lower bounds, stronger assumptions are sometimes necessary. For $d \geq 3$, let $c_d$ denote the infimum of all constants $c$ such that $d$-SAT is solvable in $2^{c n}$ time by a deterministic algorithm. The {\em exponential-time hypothesis} (ETH) then states that $c_3 > 0$, i.e., that $3$-SAT is not solvable in subexponential time.
The {\em strong exponential-time hypothesis} (SETH) additionally conjectures that the limit of the sequence $c_3, c_4, \ldots$ tends to $1$, which in particular is known to imply that the satisfiability problem for clauses of arbitrary length
is not solvable in $2^{c n}$  time for {\em any} $c < 1$~\cite{impagliazzo2009}.

\section{Backdoors into \TCNF} \label{sec:cc_2cnf}

Let $\Phi = Q_1 x_1\ldots Q_n x_n . \phi$ be a $\QBF(\CNF)$ formula.
Suppose $\phi = \phi_1 \land \phi_2$, where $\phi_1 \in \TCNF$
and $|V(\phi_2)| = k$, i.e., $V(\phi_2)$ is a CC-backdoor to $\TCNF$.
Before presenting our $\mathrm{FPT}$ algorithm we introduce some simplifying notation and formalia when branching on $b \in \{0,1\}$ for the outermost variable $x_1$, and analyzing the unit propagations caused by 
this assignment on the remaining instance.
For simplicity we will assume $\phi_1=\cnfprop{\phi_1}$ and that $\phi$ does not contain any unit-clauses, as if any such exist, we can assign values to any variable in such a clause trivially.
We first define the two $1$-\CNF formulas 
\[ u^0_{\phi_1}=\cnfprop{\phi_1\land (\neg x_1)}\setminus\phi_1 \] 
and 
\[ u^1_{\phi_1}=\cnfprop{\phi_1\land (x_1)}\setminus\phi_1, \] 
which are thus simply the variables affected by assigning $0$ or $1$ to  $x_1$. 
Note that $u^0_{\phi_1}$ and $u^1_{\phi_1}$ define partial assignments.
Formally, for $b \in \{0,1\}$ we define 
the function $U^b_\Phi \colon V_\exists(u^b_{\phi_1})\cup\{x_1\} \to \{0, 1\}$ 
as follows.
For $x_1$ the function $U^b_\Phi$ is defined as
$U^b_{\Phi}(x_1) = b$.
For all $x \in V_{\exists}(u^b_{\phi_1})$,
$U^b_{\Phi}(x) = 1$ if $u^b_{\phi_1}$ contains a unit clause $(x)$ and
$U^b_{\Phi}(x) = 0$ if $u^b_{\phi_1}$ contains a unit clause $(\lnot x)$.
Clearly, $u^0_{\phi_1}$ and $u^1_{\phi_1}$ are computable in polynomial time,
and so are $U^0_{\Phi}$ and $U^1_{\Phi}$.

Let $\dom(U^b_{\Phi}) = V_\exists(u^b_{\phi_1})\cup\{x_1\}$ denote
the domain of $U^b_{\Phi}$, i.e.,
variables affected by setting $x_1 = b$.
The basic idea  is now to carefully branch on the first variable: 
if any choice affects a backdoor variable in the \TCNF fragment, 
then we have decreased the backdoor size; 
otherwise, no branching is necessary under certain technical assumptions.

\begin{lemma}\label[lemma]{lem:2cnf-weak-key}
  Let $\Phi = \cQ . (\phi_1 \land \phi_2)$ be a $\QBF(\CNF)$ instance with $\phi_1 \in \TCNF$.
    If $\dom(U^b_{\Phi}) \cap \vars(\phi_2) = \emptyset$ and 
  $\cQ . \phi_1\repl{U^b_{\Phi}} \Leftrightarrow \top$
  for some $b \in \{0, 1\}$
  then $\Phi\repl{U^{\lnot b}_{\Phi}} \Rightarrow \Phi\repl{U^b_{\Phi}}$.
\end{lemma}
    
\begin{proof}
  Suppose $\Phi\repl{U^{\lnot b}_{\Phi}} \Leftrightarrow \top$,
  i.e., the existential player has a winning strategy $T^{\lnot b}_{\exists}$ on
  $\Phi\repl{U^{\lnot b}_{\Phi}}$.
  It suffices to show that in this case
  the existential player also has a winning strategy 
  $T^{b}_{\exists}$ on $\Phi\repl{U^b_{\Phi}}$.
  We define the strategy as follows.
  Let $\alpha : U^{0}_{\Phi} \cup U^{1}_{\Phi} \to \{0,1\}$
  be an assignment such that $\alpha(x) = U^{b}_{\Phi}(x)$
  if $x \in \dom(U^{b}_{\Phi})$ and $\alpha(x) = U^{\lnot b}_{\Phi}(x)$ otherwise,
  i.e., $\alpha$ defaults to the assignment 
  $U^{b}_{\Phi}$ and uses $U^{\lnot b}_{\Phi}$ otherwise.
  Then the existential player will follow $\alpha$
  on the variables in $\dom(\alpha)$ and reuse $T^{\lnot b}_{\exists}$
  on the remaining existential variables.
  
  Formally, the strategy $T^{b}_{\exists}$ is defined as follows.
  Consider an existential variable $x$ in the node $t$ of $T^{b}_{\exists}$.
  If $x \in \dom(\alpha)$, then choose $x = \alpha(x)$.
  Otherwise, let $\tau_t$ be the root-to-$t$ assignment in $T^{b}_{\exists}$.
  Since $\Phi\repl{U^{b}_{\Phi}}$ and $\Phi\repl{U^{\lnot b}_{\Phi}}$
  have the same set of universal variables appearing in the same order
  in the quantifier prefix,
  there is a node $t'$ in $T^{\lnot b}_{\exists}$ labeled with $x$
  such that the root-to-$t'$ assignment $\tau'_{t'}$ in $T^{\lnot b}_{\exists}$
  coincides with $\tau_{t}$ on all universal variables.
  There may be multiple such nodes $t'$, so pick an arbitrary one. 
  Assign $x = \tau'_{t'}(x)$.

  We claim that $T^{b}_{\exists}$ is a winning strategy.
  Consider a root-to-leaf assignment $\tau$.
  It suffices to show that $\tau$ satisfies every clause $c$ 
  in $\phi_1 \land \phi_2$ that is not satisfied by $U^b_{\Phi}$.
  We claim that no such clause $c$ contains a variable from $\dom(U^b_{\Phi})$.
  Indeed, if $c \in \phi_1$, then 
  $U^b_{\Phi}$ neither satisfies nor falsifies $c$
  because $\cQ. \phi_1[U^b_{\Phi}] \Leftrightarrow \top$.
  Moreover, by definition of unit propagation,
  $U^b_{\Phi}$ cannot falsify a literal in $c$
  without satisfying the other literal,
  hence it does not assign a value to either.
  Furthermore, if $c \in \phi_2$,
  then the claim holds because $\dom(U^b_{\Phi})$
  is disjoint from $V(\phi_2)$.

  By construction of $T^{b}_{\exists}$, there exists
  a root-to-leaf assignment $\tau$ in $T^{\lnot b}_{\exists})$
  such that $\tau(x) = \alpha(x)$ for all $x \in \dom(\alpha)$ and
  $\tau(x) = \tau'(x)$ otherwise.
  Moreover, every clause in $\phi_1 \land \phi_2$
  is satisfied by $U^{\lnot b}_{\Phi}$ or $\tau'$
  because $T^{\lnot b}_{\Phi}$ is a winning strategy.
  If $c$ is satisfied by $U^{\lnot b}_{\Phi}$,
  then $U^{\lnot b}_{\Phi}$ must satisfy a literal in $c$
  over a variable
  $x \in \dom(U^{\lnot b}_{\Phi}) \setminus \dom(U^{b}_{\Phi})$,
  and $\tau(x) = \alpha(x) = U^{\lnot b}_{\Phi}$, hence
  $\tau$ satisfies $c$.
  Otherwise, $c$ is not satisfied by either 
  $U^{\lnot b}_{\Phi}$ or $U^{b}_{\Phi}$
  and is satisfied by $\tau'$.
  Then $c$ does not contains any variables from $\dom(\alpha)$,
  and $\tau$ agrees with $\tau'$ on all variables in $c$,
  implying that $\tau$ satisfies $c$.
\end{proof}

We obtain the following corollary depending on the quantifier of the outermost variable.
    
\begin{corollary}\label[corollary]{cor:can_branch}
    Let $\Phi = Q_1 x_1 \ldots Q_n x_n . (\phi_1 \land \phi_2)$ be a $\QBF(\CNF)$ instance with a CC-backdoor $V(\phi_2)$ into $\TCNF$. Let $b \in \{0, 1\}$.
    If $\dom(U^b_{\Phi}) \cap \vars(\phi_2) = \emptyset$ and $\phi_1\repl{U^b_{\Phi}} \Leftrightarrow \top$, then 
    \begin{itemize}
        \item $\Phi \Leftrightarrow \Phi\repl{U^b_{\Phi}}$ if $Q_1 = \exists$, and
        \item $\Phi \Leftrightarrow \Phi\repl{U^{\lnot b}_{\Phi}}$ if $Q_1 = \forall$.
    \end{itemize}
\end{corollary}

 We solve $\Phi$ by greedily assuming an assignment to the first unassigned variable in the quantifier prefix, and use one step look-ahead combined with unit propagation.
    If an assignment directly fails, or affects the backdoor, we try the other assignment.
    If both assignments affect the backdoor variables we branch on both of them but in this case successfully reduce $k$.
    If both assignments fail, we simply answer no. For each variable in the quantifier prefix we thus use only polynomial time to do propagation, and either do not branch at all, or branch while simultaneously decreasing the backdoor size.

\begin{theorem} \label{the:2cnf_fpt}
  $\QBF(\CNF)$ is solvable in $\bigoh^*(2^k)$ time,
  where $k$ is the size of CC-backdoor to $\TCNF$.
\end{theorem}

    \begin{proof}
        Let $\Phi = \cQ. (\phi_1 \land \phi_2)$, where 
        $\phi_1 \in \TCNF$ and $k = |\vars(\phi_2)|$.
        Let $\cQ= Q_1 x_1 \ldots Q_n x_n$. We decide whether to branch on $x_1$ as follows. 
        \begin{enumerate}
            \item 
            If $\pref . \phi_1[U^1_{\Phi}]\Leftrightarrow\bot$ and $\pref . \phi_1[U^0_{\Phi}]\Leftrightarrow\bot$, then reject.
            \item 
            If $\pref . \phi_1[U^b_{\Phi}]\Leftrightarrow\top$ and $\pref . \phi_1[U^{\lnot b}_{\Phi}]\Leftrightarrow\bot$ for some $b \in \{0, 1\}$ then there is no need to branch and we continue with $\Phi[U^{b}_{\Phi}]$ if $Q_1 = \exists$ and otherwise reject.
            \item 
            If $\pref . \phi_1[U^b_{\Phi}]\Leftrightarrow\top$, $\pref . \phi_1[U^{\lnot b}_{\Phi}]\Leftrightarrow\top$ and $\dom(U^b_{\Phi}) \cap \vars(\phi_2) = \emptyset$ for some $b \in \{0, 1\}$ then we use the assignment prescribed by Corollary~\ref{cor:can_branch}.
        \end{enumerate}
        Otherwise, $\dom(U^1_{\Phi}) \cap \vars(\phi_2) \neq \emptyset$ and
        $\dom(U^0_{\Phi}) \cap \vars(\phi_2) \neq \emptyset$, so regardless of whether we branch on $x_1 = 1$ or $x_1 = 0$, the parameter $k$ decreases by at least one. Hence, if we repeat this for every  outermost variable then we in each iteration either do not branch at all, or branch on two possible values where we in each branch shrink the parameter. The complexity is thus bounded by $2^k$ and a polynomial factor, which is dominated by computing $\cnfprop{\phi_1}$ (recall Proposition~\ref{prop:prop_comp}) for every variable in the quantifier prefix, via  $U^b_{\Phi}$.%

        For correctness, step $1$ and $2$ are trivial, $3$ follows by \Cref{lem:2cnf-weak-key} and \Cref{cor:can_branch}, and otherwise we are branching and hence inherit correctness directly from the definition of evaluating QBFs.
        As these four cases are exhaustive, correctness for the full algorithm follows directly.
    \end{proof}

\begin{example} \label{ex:branching_2cnf}
We illustrate the algorithm of \Cref{the:2cnf_fpt} on the formula $\Phi$ of Example~\ref{ex:tcnf_backdoor}. The outermost variable $x_1$ is  existentially quantified. Setting $x_1=0$ propagates $(x_1 \lor x_3)$ to the unit clause $(x_3)$, hence $\dom(U^0_\Phi)=\{x_1,x_3\}$ with $U^0_\Phi(x_1)=0$ and $U^0_\Phi(x_3)=1$. Symmetrically, $x_1=1$ propagates $(\neg x_1 \lor x_4)$ to $(x_4)$, hence $\dom(U^1_\Phi)=\{x_1,x_4\}$. Both $\dom(U^0_\Phi)$ and $\dom(U^1_\Phi)$ intersect $V(\phi_2)$, so we branch on $x_1$, but in each branch the parameter $k$ strictly decreases. For instance, on the branch $x_1=0,\,x_3=1$, the clause $\phi_2$ simplifies to $(\neg x_4 \lor \neg x_5) \in \TCNF$ and the formula
\[
  \forall x_2\, \exists x_4\, \exists x_5\,.\,(x_2 \lor x_5) \land (\neg x_4 \lor \neg x_5)
\]
is in \TCNF and can be solved in polynomial time.
The branch $x_1=1$ is symmetric, so $\Phi$ evaluates to true.
\end{example}

This running time is optimal under the SETH since an asymptotically faster algorithm immediately implies a faster algorithm for $\QBF(\CNF)$ by letting the backdoor cover the entire instance.

\section{Backdoors into Affine} \label{sec:cc_affine}

We proceed by analyzing the affine case which is notably different
from \TCNF.
The main idea is to
simultaneously carefully manipulate which variables only occur in a single row while
computing a kernel via repeated applications of Gaussian
elimination. 

Let $\phi = \{(A_1, b_1), \dots, (A_m, b_m)\}$ 
be a system of affine constraints.
We define basic linear-algebraic operations
which will be used throughout the algorithm.
Suppose $A_i \neq \emptyset$ and $x \in A_i$.
By \emph{pivoting on $x$ and $i$} we denote
the operation of removing $x$ from every 
equation except $(A_i, b_i)$ and
substituting its value according to $(A_i, b_i)$.
Formally,
\[
\begin{aligned}
\pivot & (\phi,x,i)
  ={} \{(A_j,b_j)\mid j\in[m],\, x\notin A_j\} \\
  &\cup \{(A_j\triangle A_i,\, b_j\xor b_i)\mid j\in[m],\, j\neq i,\, x\in A_j\} \\
  &\cup \{(A_i,b_i)\}.
\end{aligned}
\]
where $\triangle$ denotes the symmetric difference of two sets.
We say that a variable is \emph{private} if it only occurs in one equation.
Observe that $x$ is a private variable in $\pivot(\phi, x, i)$.
It is also easy to see that
for every quantifier prefix $\cQ$ on $V(\phi)$,
we have $\cQ. \phi \Leftrightarrow \cQ. \pivot(\phi, x, i)$
because $\phi$ and $\pivot(\phi, x, i)$ have the same
set of satisfying assignments.
Moreover, if $x$ is the innermost variable in $A_i$
and it is existential, we can now \emph{eliminate} it from the system
without affecting the truth value of the formula.
We define $\elim(\phi,x,i) = \pivot(\phi, x, i) \setminus (A_i,b_i)$
and note that if $x$ is innermost in $A_i$ and existential, then
for all quantifier prefixes $\cQ$, we have
$\cQ. \phi \Leftrightarrow \cQ'. \elim(\phi,x,i)$,
where $\cQ'$ is obtained by removing $x$ from $\cQ$.
Indeed, since $x$ is the innermost variable in $(A_i, b_i)$
and private,
for any partial assignment $\alpha : A_i \setminus \{x\} \to \{0,1\}$,
the existential player can set $x$ to
$\bigoplus_{y \in A_i \setminus \{x\}}\alpha(y) \oplus b_i$
to satisfy $(A_i,b_i)$ without affecting any other equation.
    
We are now ready to introduce our first step towards a kernel.
Given a $\CNFAFF$, we will effectively remove
every equation from the \AFF part where
the innermost variable does not occur in the \CNF part.
If the backdoor set is of size $k$,
this already yields a formula with at most $k$
equations in the \AFF part.
With an additional step we ensure that every such formula
contains at most one non-backdoor variable.

\begin{definition} \label{def:kernel}
  Let $\cQ. \phi$ be in $\AFF$ and 
  $X \subseteq V(\phi)$ be a subset of variables.
  Define $\kernel(\cQ, \phi, X)$ as follows:
  \begin{enumerate}
    \item While there is an equation $(A_i,b_i)\in\phi$
      with innermost variable $x$ such that
      $x \notin X$, eliminate $x$, i.e.,
      let $\phi := \elim(\phi, x, i)$.
    \item While there is an equation $(A_i, b_i)\in\phi$
      such that $A_i \setminus X \neq \emptyset$ and
      the innermost variable in $A_i \setminus X$
      is not private, pivot on said variable, i.e.,
      let $\phi := \pivot(\phi, x, i)$.
    \item For every equation $(A, b)$ in $\phi$
      with $A \setminus X \neq \emptyset$, 
      remove every variable of $A \setminus X$ 
      except the innermost.
  \end{enumerate}
\end{definition}

Now we verify correctness of the reduction.

\begin{lemma}\label{lem:kernel}
  Let $\cQ. (\phi_1 \land \phi_2)$ be in $\CNFAFF$
  with $\phi_1 \in \AFF$ and $X = V(\phi_2)$.
  If $\cQ. \phi_1 \Leftrightarrow \top$,
  then $\cQ. (\phi_1 \land \phi_2) \Leftrightarrow \cQ. (\kernel(\cQ, \phi_1, X) \land \phi_2)$.
\end{lemma}
\begin{proof}
  First, observe that $\cQ. \phi_1 \Leftrightarrow \top$
  implies that the innermost variable in every equation
  is always existential, otherwise the universal player
  can falsify an equation solely using the innermost variable.
  Consider the first step in Definition~\ref{def:kernel}.
  If $x$ is the innermost variable in an equation $(A_i,b_i) \in \phi_1$,
  then $x$ is existential and private in $\pivot(\phi_1,x,i)$.
  Moreover, if $x \notin X$, then the value assigned to it
  does not affect $\phi_2$, so the existential player
  can use it to satisfy $(A_i,b_i)$, hence
  the equation can be safely removed.
  
  The second step in Definition~\ref{def:kernel} is pivoting,
  which preserves the set of satisfying assignments.
  To verify correctness of the third step, 
  let $\phi'_1$ be the formula obtained after steps 1 and 2.
  Consider an equation
  $(A, b) \in \phi'_1$ with $A \setminus X \neq \emptyset$.
  Let $x$ be the innermost variable in $A \setminus X$.
  Note that by the second step, $x$ is private.
  Let $A' = (A \cap X) \cup \{x\}$ and
  let $\phi''_1$ be the formula
  obtained by replacing $(A, b)$ with $(A', b)$.
  We claim that $\cQ. (\phi'_1 \land \phi_2) \Leftrightarrow \cQ. (\phi''_1 \land \phi_2)$,
  more precisely, the player that controls $x$ wins on the left-hand side
  formula if and only if they win on the right-hand side formula.
  Indeed, to obtain a strategy for one formula,
  a player that controls $x$ can modify their strategy 
  for the other formula by replacing the value of $x$ with
  $\bigoplus_{y \in A \setminus X} \tau(y)$,
  where $\tau$ is the partial assignment up to and including $x$.
  Since $x$ is a private variable, this choice
  does not affect any other equation, and since
  $x \notin X$, it does not affect $\phi_2$.
\end{proof}

We additionally want to prove that Definition~\ref{def:kernel} lets us construct an actual kernel in polynomial time.

\begin{lemma} \label{lem:cnfaff_kernel}
  There is a polynomial-time algorithm
  that takes a formula $\cQ. (\phi_1 \land \phi_2)$
  in $\CNFAFF$ with $\phi_1 \in \AFF$ and
  $|V(\phi_2)| = k$, and returns
  an instance of the same problem
  on at most $2k$ variables.
\end{lemma}
\begin{proof}
First, check if $\cQ. \phi_1 \Leftrightarrow \top$,
  and if not, output a trivially false formula,
  e.g. $(\bot)$.
  Otherwise,
  let $X = V(\phi_2)$ and let
  $\phi'_1 = \kernel(\cQ, \phi_1, X)$.
  As everything needed to do this is pivoting, 
  substituting and removing variables,  
  it can be constructed in polynomial time.
  Observe that there are at most $k$
  equations in $\phi'_1$, each of them
  containing at most $k + 1$ variables,
  at most one of which is from $V(\phi) \setminus X$.
  Hence, $\phi'_1 \land \phi_2$ contains at most $2k$ variables and we can construct a new prefix $\cQ'$ by taking $\cQ$ and removing all variables not in $\phi'_1 \land \phi_2$.
  As by \Cref{lem:kernel} we have $\cQ. (\phi_1 \land \phi_2) \Leftrightarrow \cQ'. (\phi'_1 \land \phi_2)$.
\end{proof}

Observe that the formula we obtain in~\Cref{lem:cnfaff_kernel}
contains at most $k$ equations and 
an arbitrary \CNF on $k$ variables,
which can have $\bigoh(3^k)$ clauses.
The size of the formula
cannot be reduced to polynomial in $k$
unless coNP is included in NP/poly (which implies that the polynomial hierarchy collapses)~\cite{10.1145/2629620}.

For the purposes of an $\mathrm{FPT}$ algorithm,
we could stop here and just branch in $\bigoh^*(2^{2k})$ time.
However, we can obtain the SETH-optimal running time $\bigoh^*(2^{k})$
by noting how the affine part of the instance forces assignments to certain variables.
This gives us the main theorem for the section.

\begin{theorem} \label{the:aff_fpt}
  $\QBF(\CNFAFF)$ is solvable in $\bigoh^*(2^k)$ time,
  where $k$ is the size of a CC-backdoor to $\AFF$.
\end{theorem}
\begin{proof}
By \Cref{lem:cnfaff_kernel} we can construct our kernel $\cQ'. (\phi'_1 \land \phi_2)$ such that $\cQ. (\phi_1 \land \phi_2) \Leftrightarrow \cQ'. (\phi'_1 \land \phi_2)$ in polynomial time.
By pivoting on the innermost variable of each equation in $\phi'_1$,
we make the said variables private.
Now, for any equation $(A,b)\in \phi'_1$ with innermost variable $x$, 
observe that $x$ is existential since otherwise $\cQ.\phi_1 \Leftrightarrow\bot$
and the kernel is a trivially false instance.
Thus, once every variable in $A\setminus\{x\}$ has been assigned by some partial assignment $\tau$, 
we can directly assign a value to $\tau(x)$ so that $\bigoplus_{y \in A} \tau(y) = b$.
Hence, $|\phi'_1|$ variables can be assigned directly in each branch, 
and as $|V(\phi'_1 \land \phi_2)|\leq |\phi'_1| + k$, 
we only need to branch on at most $k$ variables, 
giving $\bigoh^*(2^k)$ running time.
\end{proof}

Hence, we have proven that $\QBF(\CNFAFF)$
is in $\mathrm{FPT}$ when parameterized by the size $k$ of the affine CC-backdoor. Again, since the backdoor can be an arbitrary QBF instance, this upper bound is tight under SETH.
Note that if we do not restrict ourselves to a resulting formula strictly in $\CNFAFF$, we could from \Cref{lem:cnfaff_kernel} in polynomial time construct a formula of size $|\phi_2|$ with only $k$ variables by taking private variables from $\phi'_1$ and substituting.
I.e. for any equation $(A,b)\in \phi'_1$ with innermost private variable $x\in A$, replace $x$ by $b\xor \bigoplus_{y \in A \setminus x} y$.
However, the $2k$ variable kernel is in the worst case a more compact representation.

\section{Intractability: Backdoors to Horn} \label{sec:cc_horn}
The last classical class of QBF that we consider is $\HORN$ where we, in contrast to the positive FPT results thus far, prove that CC-backdoor evaluation is not $\mathrm{FPT}$ unless $\mathrm{FPT} = \mathrm{W}[1]$.
The starting point of our reductions is the well-known
$\mathrm{W}[1]$-hard \textsc{Multipartite Independent Set (MIS)} 
problem~(see e.g. Section~13.2~in~\cite{parameterizedalgorithmsbook}):
an instance is a graph with the vertex set partitioned into $k$ parts,
where $k$ is the parameter,
and the computational question is whether the graph contains
an independent set with one vertex from each part,
i.e., whether one can choose one vertex from each part
so that no pair of chosen vertices is adjacent.
Our reduction is inspired by Lemma~14~in~\cite{DBLP:conf/ijcai/ErikssonLOOPR24}
but modified to produce a formula with a CC-backdoor to \HORN.

\ifshort
\begin{lemma} \label{lem:mis_to_horn}
  There is a polynomial-time reduction that takes
  an instance $(G, k)$ of MIS
  and produces 
  a $QBF(\CNF)$ formula $\forall Y \exists X.\phi$ with CC-backdoor to $\HORN$ of size $k$ such that $(G, k)$ is a yes-instance
  if and only if $\pref.\phi \Leftrightarrow \bot$.
\end{lemma}
\begin{proof} (Sketch)
 Let $G$ be a graph with vertex set $V_1 \uplus \dots \uplus V_k$.
  For convenience, we refer to vertex $j$ in part $V_i$ as $(i, j)$.
  For every vertex $v = (i,j) \in V(G')$,
  add a clause $C_{v}$ to $\phi$ defined as
  \[
    C_{v} = \{y_v\} \cup \{\lnot y_u : u \in N(v)\} \cup \{\lnot x_i\}.
  \]
  Finally, add the clause $C_\exists=(x_1 \lor \ldots \lor x_k)$ to $\phi$. 
\end{proof}  
\fi

\iflong
\begin{lemma} \label{lem:mis_to_horn}
  There is a polynomial-time reduction that takes
  an instance $(G, k)$ of MIS
  and produces 
  a $QBF(\CNF)$ formula $\forall Y \exists X.\phi$ with CC-backdoor to $\HORN$ of size $k$ such that $(G, k)$ is a yes-instance
  if and only if $\pref.\phi \Leftrightarrow \bot$.
\end{lemma}
\begin{proof}
  Let $G$ be a graph with vertex set $V_1 \uplus \dots \uplus V_k$.
  For convenience, we refer to vertex $j$ in part $V_i$ as $(i, j)$.

  We will construct out formula $\forall Y \exists X. \phi$
  on variables $Y = \{ y_v : v \in V(G) \}$ and $X = \{x_1, \dots, x_k\}$
  that is false if and only if $(G, k)$ has an independent set
  with exactly one vertex from every part.
  For every vertex $v = (i,j) \in V(G')$,
  add a clause $C_{v}$ to $\phi$ defined as
  \[
    C_{v} = \{y_v\} \cup \{\lnot y_u : u \in N_G(v)\} \cup \{\lnot x_i\}
  \] where $N_G(v)$ is the set of neighbors to $v$.
  Finally, add the clause $C_\exists=(x_1 \lor \ldots \lor x_k)$ to $\phi$.
  Clearly the conjunction of all clauses $C_{v}$ is in \HORN, and
  $C_\exists$ contains $k$ variables. Hence, $\forall Y \exists
  X.\phi$ is a $QBF(\CNF)$ formula with a CC-backdoor to $\HORN$ of size $k$.
We only give a sketch of the correctness proof since it is similar the proof of Lemma~\ref{lem:ep_en_hard}. First observe that $\forall Y \exists X.\phi$ is false if and only if there is $\alpha \colon Y \to \{0,1\}$ to $Y$ such that for every $i \in [k]$
there is a vertex $v \in V_i$ whose clause $C_v$ is falsified except for $\lnot x_i$. Such $\alpha$ naturally corresponds to
choosing exactly one vertex from each part $V_i$. The construction of the
clauses then ensures that the chosen vertices are pairwise non-adjacent, since
otherwise $\alpha$ would have to assign both $y_u$ and $\lnot y_u$ to $0$ for
some vertex $u$. 
Then, such an $\alpha$ exists if and only if $G$ contains an independent set selecting exactly one vertex from each part.
\end{proof}

To reduce from $\HORN$ to $3\text{-}\HORN$ one 
can apply a  standard transformation
to every Horn clause with more than $3$ literals
until no such clause remains.

\begin{lemma} \label{lem:horn_to_3horn}
  Let $\Psi = \cQ. (\phi \land C)$ be a $\QBF(\CNF)$ formula
  where $C = \{x_1, \lnot x_2, \dots, \lnot x_r\}$
  for $r > 3$.
  Then $\Psi$ is equivalent to
  $\Psi' = \cQ \exists v. (\phi \land C_1 \land C_2)$,
  where $C_1 = \{x_1, \lnot x_2,\lnot v\}$ and
  $C_2 = \{v, \lnot x_3, \dots, \lnot x_r\}$.
\end{lemma}

Observe that the clauses $C$ and $C_2$ in 
the lemma above are in \HORN, $C_1$ is $3$-\HORN and $C_2$ contains at least one fewer variable than $C$.
Hence, we may first run the algorithm from 
\Cref{lem:mis_to_horn},
obtaining a formula from the MIS instance,
then apply \Cref{lem:horn_to_3horn} to that formula
until every \HORN clause has size at most $3$.
Thus obtaining an FPT-reduction from
MIS to the CC-backdoor evaluation problem to $3$-\HORN.
\fi

\ifshort
A standard construction allows us to reduce from HORN to 3-HORN and we obtain the following tightened result (see the technical appendix for details).
\fi

\begin{theorem} \label{the:horn_hard}
  $\QBF(\CNF)$ is $\mathrm{W}[1]$-hard parameterized by
  the CC-backdoor size to $3$-\HORN even on
  formulas with quantifier prefix of the form $\forall \dots \forall \exists \dots \exists$.
\end{theorem}

\medskip

Having established $\mathrm{W}[1]$-hardness for CC-backdoors to \HORN,
it is natural to ask whether the hardness also holds for its subclasses.
To this end, recall that $\EP$ (and its dual $\EN$) is the class of formulas that allows
positive (negative) clauses of unbounded arity, implications and negative (positive) unit clauses.

\begin{lemma} \label{lem:ep_en_hard}
  $\QBF(\CNF)$ is $\mathrm{W}[1]$-hard parameterized by
  the CC-backdoor size to $\EP$ or $\EN$ even on
  formulas with quantifier prefix of the form $\forall \dots \forall \exists \dots \exists$.
\end{lemma}
\begin{proof}
  We show the proof for $\EN$ and the argument for $\EP$ is analogous.
  The hardness follows by a modification of the proof in \Cref{lem:mis_to_horn}.
  We start with an instance $(G, k)$ of MIS and in polynomial time
  construct a formula $\Phi = \forall Y \exists X. \phi_1 \land \phi_2$
  such that $\phi_1 \in \EN$, $V(\phi_2) = X = \{x_1,\dots,x_k,z_1,\dots,z_k\}$ and 
  $\Phi \Leftrightarrow \bot$ if and only if $(G, k)$ is a yes-instance.  

  Let  $V_1 \uplus \dots \uplus V_k$ be the vertex set of $G$.
  We let $Y = \{y_v : v \in V(G)\}$.
  For a variable $v \in V_i$, define the clause
  \[
    C_v = \{\neg y_w : w \in V_i \setminus \{v\}\} \cup \{\neg y_u : u \in N_G(v) \} \cup \{\neg x_i\}.
  \]
  Observe that $C_v$ is negative.
  For a color class $i \in [k]$, let
  \[ 
    \psi_i = \bigwedge_{v \in V_i} (z_i \to y_v).
  \]
  Finally, let $\phi_1 = \bigwedge_{v \in V(G)} C_v \land \bigwedge_{i=1}^{k} \psi_i$
  and $\phi_2$ contain a single clause $\{x_1,\dots,x_k,z_1,\dots,z_k\}$.

  Observe that $\Phi \Leftrightarrow \bot$ if and only if there is
  an assignment $\alpha : Y \to \{0,1\}$ such that $\Phi[\alpha]$ is unsatisfiable.
  The latter holds if and only if 
  $\Phi[\alpha]$ implies $(\neg x_i)$ and $(\neg z_i)$ for all $i \in [k]$.
  If $\alpha(y_v) = 1$, we will say that $\alpha$ \emph{selects} $v$.
  To ensure that $\Phi[\alpha]$ implies $(\neg z_i)$, by the definition of $\psi_i$,
  is necessary and sufficient for $\alpha$ to 
  select at least one vertex in $V_i$.
  To ensure that $\Phi[\alpha]$ implies $(\neg x_i)$, it is necessary and sufficient
  for there to be a clause $C_v$ with $v \in V_i$ such that
  $\alpha$ falsifies all literals in $C_v \setminus \{\neg x_i\}$.
  Thus, if $\alpha$ selects some $v \in V_i$,
  then to make sure that $\Phi[\alpha]$ implies $(\neg x_i)$,
  $\alpha$ must not select any other vertex $w \in V_i \setminus \{v\}$
  or any neighbor $u$ of $v$ in $G$.
  Thus, the set of vertices selected by $\alpha$ must form
  an independent set in $G$ with at least one vertex
  from each color class.
\end{proof}

We note that the reduction from \Cref{lem:horn_to_3horn}
does not apply to $\EP$, so the hardness for bounded-arity
formulas remains open, which we discuss in greater detail in Section~\ref{sec:qcsp}.
To complement this hardness result, we show that disallowing implications
makes the problem fixed-parameter tractable.

\begin{lemma} \label{lem:ep_no_impl_fpt}
  Let $\cC \subseteq \EP$ be the class of formulas with
  positive clauses (of any size) and negative unit clauses.
  Then $\QBF(\CNF)$ is $\mathrm{FPT}$ parameterized by
  the CC-backdoor size to $\cC$.
\end{lemma}
\begin{proof} 
  First, we may remove unit clauses:
  if there is a unit clause with a universal variable,
  then we have a no-instance;
  if there is a unit clause with an existential variable,
  then we have assign that variable to the value
  satisfying the unit clause.
  Observe that when no unit clauses remain,
  all clauses outside the backdoor are positive.
  Now we remove non-backdoor variables:
  if a non-backdoor variable is existential, 
  assign it $1$, and if a non-backdoor
  variable is universal, assign it $0$.
  Clearly, these choices are safe for positive clauses.
  Finally, we are left with a formula with only backdoor variables,
  and it can be evaluated in time $\bigoh^*(2^k)$,
  where $k$ is the number of variables.
\end{proof}

Naturally, the result extends to the corresponding subclass of $\EN$ without implication.

\section{Towards a Full CC-Classification} \label{sec:qcsp}

\newcommand{\evaluation}[1]{\text{\textsc{$#1$-EV}}}

In this section we analyse to which extent the results so far completely describe  the complexity of the CC-evaluation problem. As we will prove, there is only one case (or two, if we also count the dual) missing for a full dichotomy, namely, bounded arity $d$-$\EP$ (and $d$-$\EN$) for $d \geq 3$.
To study this in a general setting we first define
a \emph{Boolean constraint language} $\Gamma$
as a set of Boolean relations.
Let $\CCC_{\Gamma}$ be the class of 
conjunctive formulas with atoms only using
relations from $\Gamma$. 
It is common to refer to the problem $\QBF(\CCC_{\Gamma})$
as $\QCSP(\Gamma)$ or QCSP$(\Gamma)$.
We are thus interested in the complexity
of conjunctive $\QBF$ parameterized by clause covering
backdoor size to $\CCC_{\Gamma}$ for every $\Gamma$. We let $\evaluation{\Gamma}$ be the problem of deciding $\QBF(\CNF + \Gamma)$ when given a CC-backdoor $B$ into $\CCC_{\Gamma}$ (where $\CNF + \Gamma$ is a shorthand for $\CNF + \CCC_{\Gamma}$). Naturally, we are then interested in $\evaluation{\Gamma}$ when parameterized by $|B|$. As we will prove, we are surprisingly close to a full parameterized complexity dichotomy.

To avoid exhaustive case analyses we apply the \emph{algebraic approach} where the basic idea is to equip each set of relations $\Gamma$ with an algebraic object $\pol(\Gamma)$ (\emph{polymorphisms}) such that $\pol(\Gamma)$ abstracts away from many irrelevant details, but still decides the (classical) complexity of $\QCSP(\Gamma)$. 
Concretely, an operation $f \colon \{0,1\}^d \to \{0,1\}$
is a \emph{polymorphism} of $\Gamma$ if the 
following holds for all relations $R \in \Gamma$:
for every $(d \times r)$-matrix of $A = (a_{ij})$
such that all rows $(a_{i1}, \dots, a_{ir})$ belong to $R$,
the tuple $(b_1, \dots, b_r)$ obtained by applying
$f$ to $A$ column-wise, i.e.,
$b_j = f(a_{1j}, a_{2j}, \dots, a_{dj})$,
also belongs to $R$. We let $\pol(\Gamma)$ be the set of all polymorphisms of $\Gamma$. 

\ifshort

We let $\maj(x,y,z) := (x \land y) \lor (y \land z) \lor (z \land x)$ and $\mnrty(x,y,z) = x \oplus y \oplus z$ (where $\oplus$ is the XOR operation. We obtain the following dichotomy.

\begin{theorem}($\star$) \label{the:classification}
  Let $\Gamma$ be a Boolean constraint language.
  Then $\evaluation{\Gamma}$ parameterized by the CC-backdoor size to $\Gamma$ is 
  \begin{itemize}
    \item FPT if $\maj$, $\mnrty$, $x \land (y \lor z)$, or $x \land (y \lor \overbar{z})$ are polymorphisms of $\Gamma$,
    \item W[1]-hard if $\lor$ or $\land$ are polymorphisms of 
      $\Gamma$ but $\maj$ and $x \land (y \lor z)$, $x \land (y \lor \overbar{z})$  are not,
      even on formulas with a single quantifier alternation, and
    \item para\PSPACE-hard otherwise.
  \end{itemize}
\end{theorem}
\begin{proof} (Sketch)
One can first utilize Schaefer's dichotomy theorem~\cite{sch78} for Boolean QCSP to separate the \PSPACE-complete cases to the tractable cases. Then, it is known (see, e.g., \cite{chen2006rendezvous}) that polymorphisms of constraint languages can be described by certain logical definition (\emph{primitive positive definitions}) that can be used to obtain FPT-reduction between CC-backdoor evaluation problems. This can be exploited to show that the positive cases from the preceding sections (\EP, \TCNF, \AFF) together with their duals cover all possible FPT cases.
\end{proof}
\fi 

\iflong
For a Boolean function $f \colon \{0,1\}^d \to \{0,1\}$ we define its \emph{dual} operation $f^{\mathrm{dual}}$ as $f^{\mathrm{dual}}(\overbar{x_1}, \ldots, \overbar{x_d}) = \overbar{f(x_1, \ldots, x_d)}$ for all $x_1, \ldots, x_d \in \{0,1\}$, where $\overbar{x_i}$ is Boolean complement. Similarly, define $R^{\mathrm{dual}} = \{(\overbar{x_1}, \ldots, \overbar{x_n}) \mid (x_1, \ldots, x_n) \in R\}$ for an $n$-ary relation $R$, and  $\Gamma^{\mathrm{dual}} = \{R^{\mathrm{dual}} \mid R \in \Gamma\}$ for a set of relations $\Gamma$. It is then well-known that $\{f^\mathrm{dual} \mid f \in \pol(\Gamma)\} = \pol(\Gamma^{\mathrm{dual}})$.
Last, say that a function $f$ is \emph{idempotent} if $f(x, \ldots, x) = x$ for all $x \in \{0,1\}$.  Then, it is well-known and easy to see that each function in $\pol(\Gamma^+)$, where $\Gamma^+ = \Gamma \cup \{\{(0)\}, \{(1)\}\}$, is idempotent. We then have the following basic relationships.

\begin{theorem} \label{thm:can_reduce}
    Let $\Gamma$ be a (not necessarily finite) constraint language and let $\Delta$ be a finite constraint language. The following statements are true.

    \begin{enumerate}
        \item 
    If $\pol(\Gamma) \subseteq \pol(\Delta)$ then $\evaluation{\Delta} \leq \evaluation{\Gamma}$,
        \item $\evaluation{\Gamma}$ and $\evaluation{\Gamma^{\mathrm{dual}}}$ are FPT-equivalent,
        \item 
     $\evaluation{\Gamma^+}$ and $\evaluation{\Gamma}$ are FPT-equivalent.
     \end{enumerate}
\end{theorem}

\begin{proof}
    We prove each item in turn. The first claim implies~\cite{chen2006rendezvous} that each relation in $\Delta$ is logically equivalent to a conjunctive formula (possibly containing fresh existentially quantified variables) over $\Gamma$ (where each atom either uses a relation from $\Gamma$ or is an equality constraint). Importantly, since $\Delta$ is finite, we can precompute all such definitions and store them in a table of constant size. Given $\cQ . (\phi_1 \land \phi_2)$, where $\phi_1 \in \CCC_{\Delta}$, we then replace each constraint in $\phi_1$ by the corresponding set of constraints over $\Gamma$, and treat the fresh variables as existentially quantified, added last in the prefix. We then gradually remove equality constraints by either identifying the two variables in question, or trivially detecting that the formula is false (if one variable in the equality constraint is universally quantified and is preceded by the other one). See~\cite{chen2006rendezvous} for additional details. Since this construction does not affect the backdoor constraints in $\phi_2$ the parameter $|\vars(\phi_2)|$ does not increase.

For the second claim, given $\cQ . (\phi_1 \land \phi_2)$, where $\phi_1 \in \CCC_{\Gamma}$, replace each atom $R(\cdot)$ in $\phi_1$ by $R^{\mathrm{dual}}(\cdot)$, replace each clause $(l_1 \lor \ldots \lor l_r)$ in $\phi_2$ by $(\lnot l_1 \lor \dotsb \lor \lnot l_r)$, and keep the quantifier prefix $\cQ$ intact. Any tree witnessing a yes-instance of one instance can be converted into a tree for the dual formula by complementing each value, since this preserves the truth of $\phi_1$ (by definition of the dual atoms) and of $\phi_2$ (since the literals of every clause have been negated).

    For the third claim, let $\cQ . (\phi_1 \land \phi_2)$ be an instance of $\evaluation{\Gamma^+}$. We introduce two fresh variables $X_0$ and $X_1$ and for each variable $x \in \vars(\phi_1)$ where $\{(0)\}(x) \in \phi_1$ we identify it with $X_0$. Dually, if $\{(1)\}(x) \in \phi_1$ we identify it with $X_1$, and if we detect that we have both $\{(0)\}(x), \{(1)\}(x) \in \phi_1$ then we have a trivial no-instance. Next, since $\{(0)\}(X_0)$ can be viewed as a unary clause $(\overbar{X_0})$, and  $\{(1)\}(X_1)$ as a unary clause $(X_1)$, we delete them from $\phi_1$ and instead take $\phi_2 \cup \{(X_1), (\overbar{X_0})\}$ as the new backdoor. This only increases the parameter by 2 and the reduction is indeed FPT.
\end{proof}

To have any hope of obtaining a CC-backdoor complexity classification we first need to firmly establish the tractable cases of QCSP. The following classification appeared without proof as 
Theorem 6.1 in Schaefer~\cite{sch78}, restated in modern terminology using polymorphisms (see e.g.~\cite{chen2006rendezvous}). First,
define the following operations:
\begin{itemize}
  \item binary $\min(x,y) := x \land y$,
  \item binary $\max(x,y) := x \lor y$,
  \item ternary $\maj(x,y,z) := (x \land y) \lor (y \land z) \lor (z \land x)$,
  \item ternary $\mnrty(x,y,z) = x \oplus y \oplus z$, where $\oplus$ is
    the binary exclusive OR operation.
\end{itemize}

\begin{theorem} \label{the:schaefer}
  Let $\Gamma$ be a Boolean constraint language.
  Then $\QCSP(\Gamma)$ is in P if 
  $\min$, $\max$, $\maj$ or $\mnrty$
  is a polymorphism of $\Gamma$,
  otherwise it is \PSPACE-hard.
\end{theorem}

Moreover, there is a neat connection between syntactically defined
classes of formulas (such as $\TCNF$, $\HORN$ and $\AFF$)
and tractable classes of $\QCSP$ defined by polymorphisms (see, e.g.,~\cite{blocks}).

\begin{lemma} \label{lemma:pol_vs_relation}
Let $\Gamma$ be a constraint language. If
\begin{enumerate}
    \item $\maj \in \pol(\Gamma)$ then $\pol(\TCNF) \subseteq \pol(\Gamma)$,
    \item 
    $\mnrty \in \pol(\Gamma)$ then $\pol(\AFF) \subseteq \pol(\Gamma)$,
    \item 
    $\min \in \pol(\Gamma)$ then $\pol(\HORN) \subseteq \pol(\Gamma)$,
    \item 
    $\max \in \pol(\Gamma)$ then $\pol(\DUALHORN) \subseteq \pol(\Gamma)$,
    \item 
    $x \land (y \lor z) \in \pol(\Gamma)$ then $\pol(\EN) \subseteq \pol(\Gamma)$, 
    \item 
    $x \lor (y \land z) \in \pol(\Gamma)$ then $\pol(\EP) \subseteq \pol(\Gamma)$.
\end{enumerate}
\end{lemma}

Additionally, for $d \geq 3$ define the $d$-ary \emph{threshold function} ($t_d$) as 
\[
t_d(x_1, \ldots, x_d)=
\begin{cases}
1 & \text{if } x_1 + \ldots + x_d < d, \\
0 & \text{otherwise}.
\end{cases}
\]
To simplify the forthcoming statements, say that $\Gamma$ is an $\EP$-language ($\EN$-language) if $\pol(\EP) \subseteq \pol(\Gamma)$ ($\pol(\EN) \subseteq \pol(\Gamma)$). We remark that $\Gamma$ is an $\EP$-language if and only if each $R \in \Gamma$ can be described as CNF formula where each clause is in $\EP$~\cite{CREIGNOU20081103}. %
 
 Additionally, the threshold function $t_{d}$ provides a precise description of definability of $(x_1 \lor \ldots \lor x_{d})$ in the sense that if $t_{d} \in \pol(\Gamma)$ then $\Gamma$ does not contain (nor can it define) $(x_1 \lor \ldots \lor x_{k})$ for any $k \geq d$~\cite{blocks}. 

\begin{theorem}
Let $\Gamma$ be a finite $\EP$-language. Then $\evaluation{\Gamma}$ when parameterized by the CC-backdoor size is
\begin{enumerate}
\item FPT if $x \lor (y \land \overbar{z}) \in \pol(\Gamma)$ or $t_3 \in \pol(\Gamma)$,
\item FPT-equivalent to $\evaluation{d\text{-}\EP}$ for some $d \geq 3$ depending on the arity of of $\Gamma$.
\end{enumerate}
\end{theorem}

\begin{proof}
We consider each case in turn. First, we observe that $f(x, y, z) = x \lor (y \land \overbar{z})$ does not preserve implication $\{(0,0), (0,1), (1,1)\}$ since $f((0,0), (1,1), (0,1)) = (f(0,1,0), f(0,1,1)) = (1,0) \notin {(0,0), (0,1), (1,1)}$. It follows that each relation in $\Gamma$ can be expressed as a CNF-formula where each clause is either positive, or a negative unit clause, and FPT follows from Lemma~\ref{lem:ep_no_impl_fpt}. Second, if $t_3 \in \pol(\Gamma)$ then we first observe that $t_3$ does not preserve any positive clause of arity more than 2 (applying $t_3$ to three tuples with Hamming weight $> 1$ produces the constant tuple of Hamming weight 0). It follows that each relation in $\Gamma$ can be expressed as $\TCNF$, and FPT follows from Theorem~\ref{the:2cnf_fpt}. 
For the second statement, the condition implies  that $t_{d+1} \in \pol(\Gamma)$ for some $d \geq 3$, which is known to imply that each relation in $\Gamma$ can be defined as a CNF-formula in $d$-$\EP$ and that $\pol(\Gamma) = \pol(d\text{-}\EP)$~\cite{CREIGNOU20081103}, which allows us to apply Theorem~\ref{thm:can_reduce}.
\end{proof}

An analogous statement can be made for $\EN$-languages. For languages that are not $\EP$ or $\EN$, however, we obtain a complete classification.

\begin{theorem} \label{the:classification}
  Let $\Gamma$ be a finite Boolean constraint language that is not an $\EP$- or $\EN$-language.
  Then $\evaluation{\Gamma}$ parameterized by the CC-backdoor size is 
  \begin{itemize}
    \item FPT if $\maj$ or $\mnrty$ are polymorphisms of $\Gamma$,
    \item W[1]-hard if $\min$ or $\max$ are polymorphisms of 
      $\Gamma$ but $\maj$ is not (even on formulas with a single quantifier alternation), and
    \item para\PSPACE-hard otherwise.
  \end{itemize}
\end{theorem}
\begin{proof}
For the tractable cases we (via Lemma~\ref{lemma:pol_vs_relation}) know that $\pol(\TCNF)$ or $\pol(\AFF)$ is contained in $\pol(\Gamma)$. Note that the evaluation problem for each such language is FPT by combining Theorem~\ref{the:2cnf_fpt}, Theorem~\ref{the:aff_fpt}, and  the first item of Theorem~\ref{thm:can_reduce}. 

For W[1]-hardness, we first apply Lemma~\ref{lemma:pol_vs_relation} and Theorem~\ref{the:horn_hard} and the second item of Theorem~\ref{thm:can_reduce} (to cover both \HORN and \DUALHORN). Last, by an application of the first item of Theorem~\ref{thm:can_reduce}, W[1]-hardness extends to any $\pol(\Gamma) \subseteq \pol(\HORN)$ or $\pol(\Gamma) \subseteq \pol(\DUALHORN)$ (note that the fresh existential variables are added last, so we do not increase the quantifier alternations).

Finally, if the two first condition do not hold for a finite $\Gamma$ then it is not preserved by $\maj$, $\mnrty$, $\min$, $\max$, and \Cref{the:schaefer} implies that
  the problem is \PSPACE-hard even with a backdoor of size $0$.
\end{proof}

\fi

Put together, $\evaluation{\Gamma}$ is thus either (1) FPT, (2) W[1]-hard, or (3) FPT-equivalent with $\evaluation{d\text{-}\EP}$ for some $d \geq 3$. It might also be interesting to note that the positive FPT cases hold regardless of the number of quantifier alternations. 

\section{Conclusion} \label{sec:discussion}
In this paper we introduced a variant of the well-known backdoor notion, CC-backdoors, with the goal of obtaining new $\mathrm{FPT}$ algorithms for QBF with unbounded quantifier alternations. 
We began the paper with two positive $\mathrm{FPT}$ results for \TCNF and \AFF but then quickly discovered that $\HORN$ and $\EN$ (and their duals) gave W[1]-hardness. We then managed to prove that the $d$-$\EN$ and $d$-$\EP$ ($d \geq 3$) are the only obstructions for a complete parameterized complexity classification. Hence, resolving the CC-evaluation problem for $d$-$\EP$ is an interesting, but likely challenging, open problem, and the techniques for $\TCNF$ and $\AFF$ do not seem to be applicable.
Nevertheless, there are also other unresolved questions, which we now discuss.

\paragraph{Hardness of Horn and $d$-$\EP$.}
While $\forall\exists3$-Horn being $\mathrm{W}[1]$-hard is an interesting result, it is quite far from the complete picture and a multitude of follow-up questions now arise: is $\forall\exists3$-Horn in \XP? Could it be harder than $\mathrm{W}[1]$?
By taking into consideration the definition of the polynomial hierarchy it then also becomes interesting to ask if $\forall\exists\forall\exists3$-Horn (i.e.\ the prefix is $\forall\ldots\forall\exists\ldots\exists\forall\ldots\forall\exists\ldots\exists$) is $\mathrm{W}[2]$-hard? How about other types of fixed prefixes? And then finally, is $3$-Horn with no prefix restrictions para-$\mathrm{NP}$-hard when parameterized by a CC-backdoor?

\paragraph*{Extensions to quantified constraint satisfaction.}
The QCSP problem can be defined over arbitrary finite domain. Can our Boolean results be generalized to this (significantly more general) setting? This
seems difficult to answer in full generality since the classical complexity of such
QCSPs has not been determined, but there is nothing preventing us from
studying specific languages. For example, QCSPs over linear equations
are tractable and could form suitable backdoor classes, while the
algebraic extension of \TCNF to arbitrary finite domains via {\em
  semilattice operations} is not always
tractable~\cite{DBLP:conf/cp/Chen04}. We also mention the
$0$/$1$/all-constraints~\cite{DBLP:journals/ai/CooperCJ94} as a
promising generalization of \TCNF that supports propagation. Also note
that Boolean QCSPs are always either tractable or \PSPACE-complete, but
finite-domain QCSPs exhibit a much wider range in complexity and can
e.g. also be $\mathrm{NP}$-complete, $\mathrm{coNP}$-complete,
$\mathrm{DP}$-complete, or $\Pi^P_2$-complete; however, curiously,
there is nothing between $\Pi^P_2$ and
\PSPACE~\cite{DBLP:conf/focs/Zhuk24}. Attacking these restricted QCSP
classes with a backdoor approach is tempting since, from an
algorithmic standpoint, they are likely easier to handle than \PSPACE-complete QCSPs.

\section*{Acknowledgements}

The second author was partially supported by 
the Swedish Research Council (VR)
under grant 2021-04371 and 2025-04487.
The fourth author was supported by VR
under grant 2024-00274.

\bibliographystyle{abbrv}
\bibliography{references}

\end{document}